\documentclass[a4paper,UKenglish]{lipics-v2016}
 
\usepackage{microtype}

\usepackage{todonotes}
\usepackage{xspace,xargs}   

\usepackage{url}

\title{Detecting Mutations by eBWT\footnote{This work was partially supported by the project MIUR-SIR CMACBioSeq (``Combinatorial methods for analysis and compression of biological sequences'') grant n.~RBSI146R5L.}}
\titlerunning{Detecting Mutations by eBWT} 

\author[1]{Nicola Prezza}
\author[1,3]{Nadia Pisanti}
\author[2]{Marinella Sciortino}
\author[1]{Giovanna Rosone\thanks{Corresponding author: Giovanna Rosone, giovanna.rosone@unipi.it}}
\affil[1]{University of Pisa, Dipartimento di Informatica, Italy
  }
\affil[2]{University of Palermo, Dipartimento di Matematica e Informatica, Italy
}
\affil[3]{ERABLE Team INRIA, France}
\authorrunning{N. Prezza, N. Pisanti, M. Sciortino, G. Rosone}

\Copyright{Nicola Prezza and Nadia Pisanti and Marinella Sciortino and  Giovanna Rosone}



\EventEditors{John Q. Open and Joan R. Acces}
\EventNoEds{2}
\EventLongTitle{42nd Conference on Very Important Topics (CVIT 2016)}
\EventShortTitle{CVIT 2016}
\EventAcronym{CVIT}
\EventYear{2016}
\EventDate{December 24--27, 2016}
\EventLocation{Little Whinging, United Kingdom}
\EventLogo{}
\SeriesVolume{42}
\ArticleNo{23}

\def\bigR{\mathcal{R}\xspace}
\def\bigS{\mathcal{S}\xspace}
\def\bigT{\mathcal{T}\xspace}

\def\BWT{{\rm BWT}\xspace}
\def\EBWT{{\rm eBWT}\xspace}
\def\LCP{{\rm LCP}\xspace}

\def\GSA{{\rm GSA}\xspace}
\def\EGSA{{\rm Egsa}\xspace}

\def\PosCluster{\rm{e}{\sc BWTclust}\xspace}
\def\DiscoSnp{{\sc DiscoSnp}\xspace}
\def\DiscoSnpPiu{{\sc DiscoSnp}\rm{++}\xspace}

\def\ebwt{{\normalfont\textsf{ebwt}}\xspace}

\def\lcp{{\normalfont\textsf{lcp}}\xspace}

\def\gsa{{\normalfont\textsf{gsa}}\xspace}

\newtheorem{proposition}[theorem]{Proposition}

\newcommandx{\todogio}[2][1=]{\todo[linecolor=red,backgroundcolor=red!25,bordercolor=red,#1]{#2}}

\newcommandx{\todonic}[2][1=]{\todo[linecolor=blue,backgroundcolor=blue!25,bordercolor=blue,#1]{#2}}

\newcommandx{\todomari}[2][1=]{\todo[linecolor=green,backgroundcolor=green!25,bordercolor=green,#1]{#2}}

\begin{document}

\maketitle

\begin{abstract}
In this paper we develop a theory describing how the extended Burrows-Wheeler Transform (\EBWT) of a collection of DNA fragments tends to cluster together the copies of nucleotides sequenced from a genome $G$. Our theory accurately predicts how many copies of any nucleotide are expected inside each such cluster, and how an elegant and precise \LCP array based procedure can locate these clusters in the \EBWT.
	
Our findings are very general and can be applied to a wide range of different problems. In this paper, we consider the case of alignment-free and reference-free SNPs discovery in multiple collections of reads. We  note that, in accordance with our theoretical results, SNPs are clustered in the \EBWT of the reads collection, and  we develop a tool finding SNPs with a simple scan of the \EBWT and \LCP arrays. 
Preliminary results show that our method requires much less coverage than state-of-the-art tools while drastically improving precision and sensitivity.



\end{abstract}

\section{Introduction}


The cheap and fast Next Generation Sequencing (NGS) technologies are producing huge amounts of data that put a growing pressure on data structures designed to store raw sequencing information, as well as on efficient analysis algorithms: collections of billion of DNA fragments (reads) need to be efficiently indexed for downstream analysis. 
After a sequencing experiment, the most traditional analysis pipeline begins with an error-prone and lossy mapping of the collection of reads onto a reference genome. Among the most widespread tools to align reads on a reference genome we can mention BWA \cite{LiDurbin09}, Bowtie2 \cite{LangmeadSalzberg12}, SOAP2 \cite{LiYLLYKW09}. These methods share the use of the FM-index \cite{FerraginaM00}, an indexing machinery based on the Burrows-Wheeler Transform (BWT) \cite{bwt94}.
Other approaches \cite{KKbmc15,KKbioinf15} combine an index of the reference genome with the BWT of the reads collection in order to boost efficiency and accuracy. 
In some applications, however, aligning reads on a reference genome presents  limitations mainly due to the difficulty of mapping highly repetitive regions, especially in the event of a low-quality reference genome (not to mention the cases in which the reference genome is not even available).

For this reason, indices of reads collections have also been suggested as a lossless dictionary of sequencing data, where sensitive analysis methods can be directly applied without mapping the reads to a reference genome (thus without needing one), nor assembling \cite{kissnp,kissplice12,Leggett2014,Iqbal-CORTEX}. 
In \cite{Dolle2017} the BWT, or more specifically its extension to string collections (named \EBWT \cite{MantaciRRS07,BauerCoxRosoneTCS2013}), is used to index reads from the 1000 Genomes Project \cite{1000GP} in order to support $k$-mer search queries. An \EBWT-based compressed index of sets of reads has also been suggested as a basis for both RNA-Seq \cite{CoxJakobiRosoneST2012} and metagenomic \cite{Ander2013} analyses.
There exist also suffix array based data structures devised for indexing reads collections: the $Gk$ array \cite{PhilippeSLLCR11,ISBRA-2013} and the PgSA \cite{KGD15}. 
Finally, there are several tools (\cite{kissnp,kissplice12,takeabreak14,bubble-spire,Leggett-2013,Leggett2014,Iqbal-CORTEX}) that share the idea of using the de Bruijn graph (dBG) of the reads' $k$-mers. 
The advantages of dBG-based indices include allowing therein the characterization of several biologically-interesting features of the data as suitably shaped and sized \emph{bubbles}\footnote{A \emph{bubble} in a graph is a pair of disjoint paths sharing the same source node and target node.} (e.g. SNPs, INDELs, Alternative Splicing events on RNA-Seq data, sequencing errors can all be modeled as bubbles in the dBG of sequencing data  \cite{kissnp,kissplice12,takeabreak14,bubble-spire,Leggett-2013}). The drawback of these dBG representation, as well as those of suffix array based indices such as the ones introduced in \cite{PhilippeSLLCR11,ISBRA-2013}, is the lossy aspect of getting down to $k$-mers rather than representing the actual whole collection of reads. Also \cite{KKbmc15,KKbioinf15} have this drawback as they index $k$-mers. 
An \EBWT-based indexing method for reads collections, instead, has the advantages to be easy to compress and, at the same time, lossless: (e)BWT indexes support querying $k$-mers without the need to build different indexes for different values of $k$.

Here we introduce the \emph{Positional Clustering} framework: an \EBWT-based index of reads collections where we give characterizations of (i) positions sharing a context as clusters in the \EBWT, and (ii) the onset of these clusters by means of the LCP. 
This clustering allows to locate and investigate, in a lossless index of reads collections, genome positions possibly equivalent to bubbles in the dBG \cite{takeabreak14,kissnp} \emph{independently} from the k-mer length (a major drawback of dBG-based strategies). We thus gain the advantages of dBG-based indices while maintaining those of (e)BWT-based ones. Besides, the \EBWT index also contains abundance data (useful to distinguish errors from variants, as well as distinct variant types) and does not need the demanding read-coherency check at post processing as no micro-assembly has been performed. With these promising advantages with respect to dBG-based strategies, the positional clustering framework allows likewise reference-free and assembly-free detection of SNPs \cite{kissnp,discoSNP2017bioRxiv,Iqbal-CORTEX}, small INDELs \cite{SOAPindel,Iqbal-CORTEX}, sequencing errors \cite{SR2014,SWRU17,LFP17}, 
alternative splicing events \cite{kissplice12}, rearrangements breakpoints \cite{takeabreak14} on raw reads collections.

As a proof-of-concept, we test our theoretical framework with a prototype tool named \PosCluster designed to detect positional clusters and post-process them for assembly-free and reference-free SNPs detection directly on the \EBWT of reads collection. Among several reference-free SNPs finding tools in the literature \cite{kissnp,discoSNP2015,discoSNP2017bioRxiv,Iqbal-CORTEX}, 
the state-of-the-art is represented by the well documented and maintained {\sc KisSNP} and \DiscoSnp suite \cite{kissnp,discoSNP2015,kSNP2biblio}, where \DiscoSnpPiu \cite{discoSNP2017bioRxiv} is the latest and best performing tool. 
In order to validate the accuracy of positional clustering for finding SNPs, we compared \DiscoSnpPiu sensitivity and precision 
to those of our  prototype \PosCluster. Preliminary results on human chromosomes show that, even when using relatively low coverages (22x), our tool is able to find $91\%$ of all SNPs (vs $70\%$ of \DiscoSnpPiu) with an accuracy of $98\%$ (vs $94\%$ of \DiscoSnpPiu). 

\section{Preliminaries}\label{sec:ADS_collection}
	
Let $\Sigma =\{c_1, c_2, \ldots, c_\sigma\}$ be a finite ordered alphabet with $c_1< c_2< \ldots < c_\sigma$, where $<$ denotes the standard lexicographic order.
For $s\in \Sigma^*$, we denote its letters by $s[1], s[2],\ldots,s[n]$, where $n$ is the \emph{length} of $s$, denoted by $|s|$.
We append to $s\in \Sigma^*$ an end-marker symbol $\$$ that satisfies $\$$ $< c_1$.
Note that, for $1 \leq i \leq n$, $s[i]\in \Sigma$ and $s[n+1]=\$$ $\notin \Sigma$.
A \emph{substring} of $s$ is denoted as $s[i,j] = s[i] \cdots s[j]$, with $s[1,j]$ being called a \emph{prefix} and $s[i,n+1]$ a \emph{suffix} of $s$.
	
We denote by $\bigS=\{R_1,R_2,\ldots,R_{m}\}$ a collection of $m$ strings (reads), and by $\$_i$ the end-marker appended to $R_i$ (for $1\leq i \leq m$), with $\$_i<\$_j$ if $i<j$. 
Let us denote by $P$ the sum of the lengths of all strings in $\bigS$.
 The \emph{generalized suffix array} \GSA of the collection $\bigS$ (see \cite{Shi:1996,CGRS_JDA_2016,Louza2017}) is an array containing $P$ pairs of integers $(j,r)$, corresponding to the lexicographically sorted suffixes $R_{r}[j, |R_r|+1]$, where $1 \leq j \leq |R_r|+1$ and $1 \leq r \leq m$. In particular, $\gsa(\bigS)[i]=(j,r)$ (for $1 \leq i \leq P$) if the suffix $R_{r}[j, |R_r|+1]$ is the $i$-th smallest suffix of the strings in $\bigS$. Such a notion is a natural extension of the suffix array of a string (see \cite{Manber:1990}). 
%
%
The Burrows-Wheeler Transform (\BWT) \cite{bwt94}, a well known text transformation largely used for data compression and self-indexing compressed data structure, has also been extended to a collection $\bigS$ of strings (see \cite{MantaciRRS07}). Such an extension, known as \emph{extended Burrows-Wheeler Transform} (\EBWT) or multi-string \BWT, is a  reversible transformation that produces a string that is a permutation of the letters of all strings in $\bigS$: $\ebwt(\bigS)$ is obtained by concatenating the symbols cyclically preceding each suffix in the list of lexicographically sorted suffixes of all strings in $\bigS$.   The \EBWT applied to $\bigS$ can also be defined in terms of the generalized suffix array of $\bigS$ (\cite{BauerCoxRosoneTCS2013}): if $\gsa(\bigS)[i]=(t,j)$ then $\ebwt(\bigS)[i] = R_j[t-1]$; when $t=1$, then $\ebwt(\bigS)[i]=\$_j$.   
	For $\gsa$,$\ebwt$, and $lcp$, the \emph{LF} mapping (resp. \emph{FL}) is a function that associates to each (e)\BWT symbol the position preceding (resp. following) it on the text.
	
 

The \emph{longest common prefix} (\LCP) array of a collection $\bigS$ of strings (see \cite{CGRS_JDA_2016,Louza2017,Manzini2017}), denoted by $\lcp(\bigS)$, is an array storing the length of the longest common prefixes between two consecutive suffixes of $\bigS$ in lexicographic order. 
	For each $i=2, \ldots, P$, if $\gsa(\bigS)[i-1]=(p_1,p_2)$ and $\gsa(\bigS)[i]=(q_1,q_2)$, $\lcp(\bigS)[i]$ is the length of the longest common prefix of suffixes starting at positions $p_1$ and $q_1$  of the strings $R_{p_2}$ and $R_{q_2}$, respectively. We set $\lcp(\bigS)[1]=0$.
	
For $\gsa$, $\ebwt$, and $\lcp$, the set $\bigS$ will be omitted when clear from the context.

\section{eBWT positional clustering}
\label{sec:ebwt-clust}

Let $R$ be a read sequenced from a genome $G[1,n]$. We say that $R[j]$ is a \emph{read-copy} of $G[i]$ iff $R[j]$ is copied from  $G[i]$ during the sequencing process (and then possibly changed due to sequencing errors).
Let us consider the \EBWT of a set of reads $\{R_1,\ldots, R_m\}$ of length\footnote{For simplicity of exposition, here we assume that all the reads have the same length $r$. With little more effort, it can be shown that our results hold also when $r$ is the \emph{average} read length.} $r$, sequenced from a genome $G$. Assuming that $c$ is the \emph{coverage} of $G[i]$, let us denote with $R_{i_1}[j_1], \dots, R_{i_c}[j_c]$ the $c$ read-copies of $G[i]$. Should not there be any sequencing error, if we consider $k$ such that the genome fragment $G[i+1,i+k]$ occurs only once in $G$ (that is, nowhere else than right after $G[i]$) and if $r$ is large enough so that  with high probability each $R_{i_t}[j_t]$ is followed by at least $k$ nucleotides, then we observe that the $c$ read copies of $G[i]$ would appear contiguously in the \EBWT of the reads.
We call this phenomenon \emph{\EBWT positional clustering}. 
Due to sequencing errors, and to the presence of repetitions with mutations in real genomes, a \emph{clean} \EBWT positional clustering is not realistic. However, in this section we show that, even in the event of sequencing errors, in the \EBWT of a collection of reads sequenced from a genome $G$, the read-copies of $G[i]$ still \emph{tend to be clustered} together according to a suitable Poisson distribution. 

We make the following assumptions: (i) the sequencing process is uniform, i.e. the positions from where each read is sequenced are uniform and independent random variables, (ii) the probability $\epsilon$ that a base is subject to a sequencing error is a constant\footnote{We assume this to simplify the theoretical framework. In Section~\ref{sec:snp} we will see that our framework works even on real data with sequencing simulated using realistic errors distribution.}, (iii) a sequencing error changes a base to a different one uniformly (i.e. with probability $1/3$ for each of the three possible variants), and (iv) the number $m$ of reads is large (hence, in our theoretical analysis we can assume $m\rightarrow\infty$). 

\begin{definition}[\EBWT cluster]\label{def_ebwt cluster}
	The \emph{\EBWT cluster of $i$}, with $1\leq i \leq n$ being a position on $G$, is the substring $\ebwt[a,b]$ such that $\gsa[a,b]$ is the range of read suffixes prefixed by $G[i+1,i+k]$, where $k<r$ is the smallest value for which $G[i+1,i+k]$ appears only once in $G$. If no such value of $k$ exists, we take $k=r-1$ and say that the cluster is \emph{ambiguous}.
\end{definition}

If no value $k<r$ guarantees that $G[i+1,i+k]$ appears only once in $G$, then the \EBWT cluster of $i$ does not contain only read-copies of $G[i]$ but also those of other $t-1$ characters $G[i_2], \dots, G[i_t]$. We call $t$ the \emph{multiplicity} of the \EBWT cluster. Note that $t=1$ for non-ambiguous clusters.

\begin{theorem}[\EBWT positional clustering]\label{th:eBWT clustering}
	Let $R_{i_1}[j_1], \dots, R_{i_c}[j_c]$ be the $c$ read-copies of $G[i]$. An expected number $X\leq c$ of these read copies will appear in the	\EBWT cluster $\ebwt[a,b]$ of $i$, where  $X\sim Poi(\lambda)$ is a Poisson random variable with mean
	$$
	\lambda = m\cdot  \frac{r-k}{n}\left(1-\epsilon\right)^{k}
	$$
	and where $k$ is defined as in Definition \ref{def_ebwt cluster}.
\end{theorem}
\begin{proof}
	The probability that a read covers $G[i]$ is $r/n$. However, we are interested only in those reads such that, if $R[j]$ is a read-copy of $G[i]$, then the suffix $R[j+1,r+1]$ contains at least $k$ nucleotides, i.e. $j \leq r-k$. In this way, the suffix $R[j+1,r+1]$ will appear in the \GSA range $\gsa[a,b]$ of suffixes prefixed by $G[i+1, i+k]$ or, equivalently, $R[j]$ will appear in $\ebwt[a,b]$. The probability that a random read from the set is uniformly sampled from such a position is $(r-k)/n$. If the read contains a sequencing error inside $R[j+1,j+k]$, however, the suffix $R[j+1,r+1]$ will not appear in the \GSA range $\gsa[a,b]$. The probability that this event does not happen is $(1-\epsilon)^k$. Since we assume that these events are independent, the probability of their intersection is therefore 
	$$Pr(R[j] \in \ebwt[a,b]) = \frac{r-k}{n}\left(1-\epsilon\right)^{k}$$
	This is a Bernoullian event, and the number $X$ of read-copies of $G[i]$ falling in $\ebwt[a,b]$ is the sum of $m$ independent events of this kind. Then, $X$ follows a Poisson distribution with mean $\lambda = m\cdot \frac{r-k}{n}\left(1-\epsilon\right)^{k}$. 
\end{proof}

Theorem \ref{th:eBWT clustering} states that, if there exists a value $k<r$ such that $G[i+1, i+k]$ appears only once in $G$ (i.e. if the cluster of $i$ is not ambiguous), then $X$ of the $b-a+1$ letters in $\ebwt[a,b]$ are read-copies of $G[i]$. The remaining $(b-a+1)-X$ letters are noise introduced by suffixes that mistakenly end up inside $\gsa[a,b]$ due to sequencing errors. It is not hard to show that this noise is extremely small under the assumption that $G$ is a uniform text; we are aware that this assumption is not realistic, but we will experimentally show in Section~\ref{sec:snp} that also on real genomes we can safely assume $X = b-a+1$ without affecting results. 

Note that the expected coverage of position $G[i]$ is also a Poisson random variable, with mean $\lambda ' = \frac{mr}{n}$ equal to the average coverage. On expectation, the size of non-ambiguous \ebwt clusters is thus $\lambda/\lambda' = \frac{(r-k)(1-\epsilon)^k}{r} <1 $ times the average coverage. E.g., with $k=16$, $\epsilon=0.0012$ (the study~\cite{schirmer2016illumina} reports this maximum average substitution rate for Illumina HiSeq platforms), and $r=100$ the expected cluster size is  $100\cdot\lambda/\lambda' \approx 84\%$ the average coverage. 

Finally, it is not hard to prove, following the proof of Theorem \ref{th:eBWT clustering}, 
that in the general case with multiplicity $t\geq 1$
the expected cluster size follows a Poisson distribution with mean $t\cdot \lambda$ (because the read-copies of $t$ positions are clustered together). This observation will allow us to detect and discard ambiguous clusters  using a significance test. 

So far, we have demonstrated the \EBWT positional clustering property but we don't have a way for identifying the \EBWT clusters. 
A naive strategy could be to fix a value of $k$ and define clusters to be ranges of $k$-mers in the \gsa. This solution, however, fails to separate read suffixes differing after $k$ positions (this is, indeed, a drawback of all $k$-mer-based strategies).
The aim of Theorem \ref{theorem:LCP} is precisely to fill this gap, allowing us to move from theory to practice. Intuitively, we show that clusters lie between local minima in the \LCP array. This strategy automatically detects the value $k$ satisfying Definition \ref{def_ebwt cluster} in a data-driven approach.

Our result holds only if two conditions on the \EBWT cluster under investigation are satisfied (see Proposition \ref{prop:condition 2 prob} for an analysis of the success probability).
More specifically, we require that the \EBWT cluster $\ebwt[a,b]$ of position $i$ satisfies:

\begin{enumerate}
	\item[(1)] The cluster does not have noise, i.e. $X = b-a+1$, and
	\item[(2)] Let $(p_1,j_1), (p_2,j_2) \in \gsa[a,b]$. For any $x$ such that $k\leq x<r$, if both $R_{p_1}[j_1,r]$ and $R_{p_2}[j_2,r]$ contain their leftmost sequencing errors in $R_{p_1}[j_1+x]$ and $R_{p_2}[j_2+x]$, then $R_{p_1}[j_1+x] \neq R_{p_2}[j_2+x]$.
\end{enumerate}

Proposition~\ref{prop:condition 2 prob} will show that with probability high enough Condition (2) is satisfied in practice. E.g., with $r=100$, $\epsilon=0.0012$ (see~\cite{schirmer2016illumina}), mean coverage $\lambda'=44$ (the coverage of one of our experiments in Section \ref{sec:results}), and for any $k\geq 11$, Proposition \ref{prop:condition 2 prob} shows that the condition holds (in expectation) on at least $93\%$ of the non-ambiguous clusters. 

\begin{theorem}\label{theorem:LCP}
	Let $\ebwt[a,b]$ be the \EBWT cluster of a position $i$ meeting Conditions (1) and (2). Then, there exists a value $a < p \leq b$ such that $\lcp[a+1,p]$ is a non-decreasing sequence and  $\lcp[p+1,b]$ is a non-increasing sequence.
\end{theorem}
\begin{proof}
Let us denote by $p_M$ the largest index in $(a,b]$ such that $\lcp[p_M]=M$, where $M$ is the maximum value of LCP in $(a,b]$ (if $M$ occurs multiple times, take the rightmost occurrence). We claim the theorem holds for $p=p_M$. Let us denote by $q$ and $j$, $1\leq q \leq m$, $1\leq j \leq r$, the positive integers such that $\gsa[p_M]=(j,q)$. This means, by using Condition (2), that the read $R_q$ contains the longest prefix $R_q[j,j+M-1]$ without sequencing errors, $j+M\leq r+1$. 
Consider any other suffix $R_u[j_u,r+1]$ in the range and let $R_u[j_u+x]$ be the leftmost mismatch letter in $R_u[j_u,r+1]$, with $x\geq k$.  Note that if $R_u[j_u,r+1]$ does not contain sequencing errors, then $j_u+x-1=r$, otherwise $R_u[j_u+x]$ has been mutated with an error. We suppose that the mutation at position $j_u+x$ generated a letter lexicographically smaller than that of the genome (the other case is symmetric). By Condition (2), $x\leq M+1$ and  
no other suffix $R_{v}[j_v,r+1]$ in the range satisfies $R_{v}[j_{v}+x] = R_u[j_u+x]$. 
Then, $R_u[j_u,r+1]$ falls right after a suffix $R_{u'}[j_{u'},r+1]$ such that either $R_{u'}[j_{u'},r]$ properly prefixes $R_u[j_u,r+1]$ or the leftmost mismatch occurs at position $x'\leq x$. Similarly, $R_u[j_u,r+1]$ falls right before a suffix $R_{u''}[j_{u''},r+1]$ such that $R_{u}[j_{u}+x] < R_{u''}[j_{u''}+x]$ and whose leftmost mismatch position $x''$ satisfies $x\leq x''\leq M+1$. This shows that, before suffix $R_q[j,r+1]$, other suffixes are ordered by increasing position of their leftmost mismatch letter (since $x'\leq x \leq x''$) and, then, by the lexicographic order among mismatch letters, which in particular implies that before suffix $R_q[j,r+1]$ the \lcp values are non-decreasing. Symmetrically, with a similar reasoning one can easily prove that after suffix $R_q[j,r+1]$ the \lcp values are non-increasing. 
\end{proof}

\begin{proposition}\label{prop:condition 2 prob}
	Given a \EBWT cluster $\ebwt[a,b]$ with multiplicity $t\geq 1$, Condition (2) holds with probability at least
	$$
	CDF(t\lambda,\lceil t\lambda\rceil+\delta) \cdot\left(\sum_{e=0}^3 {\lceil t\lambda\rceil+\delta\choose e}\cdot (1-\epsilon)^{\lceil t\lambda\rceil+\delta-e}\cdot \epsilon^e\cdot c_e \right)^{r-k}
	$$
for any integer $\delta\geq 0$ and where: $k$ is the value defined in Definition \ref{def_ebwt cluster}, $\lambda$ is the mean of the Poisson distribution of Theorem \ref{th:eBWT clustering}, $CDF(\mu,z)$ is the cumulative distribution function of a Poisson random variable with mean $\mu$ evaluated in $z$, and  $c_0 = c_1 = 1$, $c_2=2/3$, $c_3=2/9$.
\end{proposition}
\begin{proof}
First note that, by Theorem \ref{th:eBWT clustering},
the number of suffixes in the cluster sharing their first $k$ bases is at most $\lceil t\lambda\rceil+\delta$ with probability $CDF(t\lambda ,\lceil t\lambda\rceil +\delta)$.
We first analyze the probability that the condition holds for a fixed offset $x$ (i.e. looking at the $(x+1)$-th base of all suffixes in the range sharing their first $x$ characters), and then compute the joint probability for all $k\leq x<r$.

Note that the condition cannot fail if the number $e$ of bases $R_{p_{h}}[j_h+x]$ that have been subject to errors is either $e=0$ or $e=1$. The condition does not fail also if we have $e=2$ or $e=3$ \emph{distinct} errors, while it always fails for $e\geq 4$ (since at least two errors will produce the same base).
On a generic number $Y\leq \lceil t\lambda\rceil+\delta$ of suffixes (those sharing their first $x$ bases), one can verify that the probability that the condition does not fail for a fixed number $e<4$ of errors is ${Y\choose e}\cdot (1-\epsilon)^{Y-e}\cdot \epsilon^e\cdot c_e$, where $c_0 = c_1 = 1$, $c_2=2/3$, $c_3=2/9$.
Since these events are disjoint, we can sum these probabilities to get a lower bound to the probability that condition (2) holds on a specific offset $x$ on $Y$ suffixes. Since this probability decreases as $Y$ increases, we get a lower bound by taking the maximum $Y=\lceil t\lambda\rceil+\delta$, and obtain probability $\sum_{e=0}^3 {\lceil t\lambda\rceil+\delta\choose e}\cdot (1-\epsilon)^{\lceil t\lambda\rceil+\delta-e}\cdot \epsilon^e\cdot c_e$. 

To get a lower bound to the probability that the condition holds \emph{simultaneously} on all $k\leq x<r$, we take the product of these probabilities (note that errors at different offsets are independent events). 
This reasoning holds only if there are at most $\lceil t\lambda\rceil+\delta$ suffixes sharing their first $k$ bases --- which happens with probability $CDF(t\lambda,\lceil t\lambda\rceil+\delta)$ --- so we must include this multiplicative correction factor to get our final lower bound. 
\end{proof}

Note: to get a lower bound that is independent from $\delta$ in Proposition \ref{prop:condition 2 prob}, use the $\delta$ that maximizes the expression (since the lower bound holds for any $\delta$). 

According to Theorem \ref{theorem:LCP}, clusters are delimited by local minima in the \LCP array of the read set. This gives us a strategy for finding clusters that is \emph{independent} from $k$. Importantly, the proof of Theorem \ref{theorem:LCP} also gives us the suffix in the range (the $p$-th suffix) whose longest prefix without sequencing errors is maximized. This will be useful in the next section to efficiently compute a consensus of the reads in the cluster. 


Observe that by applying Theorem \ref{theorem:LCP} we also find ambiguous clusters. However, the expected length of these clusters is a multiple of $\lambda$, so they can be reliably discarded with a significance test based on the Poisson distribution of Theorem \ref{th:eBWT clustering}.

\section{Experimental Validation: Reference-Free SNPs Discovery}
\label{sec:snp}

When the reads dataset contains variations (e.g. two allele of the same individual, or two or more distinct individuals, or different isoforms of the same gene in RNA-Seq data, or different reads covering the same genome fragment in a sequencing process, etc.), the \EBWT positional clustering described in the previous section can be used to detect, directly from the raw reads (hence, without assembly and without the need of a reference genome), positions $G[i]$ exhibiting possibly different values, but followed by the same context: they will be in a 
cluster delimited by \LCP minima and containing possibly different letters (corresponding to the read copies of the variants of $G[i]$ in the read set). 
This general idea
can be used in several applications: error correction, assembly (the strategy finds overlaps between reads, so it can be used for assembly purposes), haplotype discovery (if fed with reads from just one diploid sample, this strategy can find heterozygous sites), and so on.

In this section, with the purpose of experimentally validating the theoretical framework of Section~\ref{sec:ebwt-clust}, we describe a new alignment-free and reference-free method that, with a simple scan of the \EBWT and \LCP array, detects SNPs in read collections.

Since (averagely) half of the reads comes from the forward (F) strand, and half from the reverse-complement (RC) strand, we denote with the term \emph{right} (resp. \emph{left}) \emph{breakpoint} those variants found in a cluster formed by reads coming from the F (resp. RC) strand, and therefore sharing the right (resp. left) context adjacent to the variant. A \emph{non-isolated SNP} \cite{discoSNP2015} is a variant at position $i$ such that the closest variant is within $k$ bases from $i$, for some fixed $k$ (we use $k=31$ in our validation procedure, see below).
The SNP is \emph{isolated} otherwise. Note that, while isolated SNPs are found twice with our method (one as a right breakpoint and one as a left breakpoint), this is not true for non-isolated SNPs: variants at the sides of a group of non-isolated SNPs are found as either left or right breakpoint, while SNPs \emph{inside} the group will be found with positional clustering plus a partial local assembly of the reads in the cluster. In the next two subsections we give all the details of our strategy. 

\subsection{Pre-processing (\EBWT computation)}
\label{subsec:Pre-processing}

Since we do not aim at finding matches between corresponding pairs of clusters on the forward and reverse strands, we 
augment the input adding the reverse-complement of the reads: for a reads set $\bigS$, we 
add $\bigS^{RC}$ as well. 
Hence, given two reads sets $\bigS$ and $\bigT$, in the pre-processing phase we compute
$\ebwt(\bigR)$, $\lcp(\bigR)$, and $\gsa(\bigR)$,
for $\bigR = \{ \bigS \cup \bigS^{RC} \cup \bigT \cup \bigT^{RC}\}$ \cite{Louza2017,CGRS_JDA_2016,LouzaGT17b}.
%
We also compute $\gsa(\bigR)$ because we will need it (see Subsection \ref{ssec:post}) to extract left and right contexts of the SNP. Though this could be achieved by performing (in external memory) multiple steps of 
LF- and FL-mappings on the \EBWT, this would significantly slow-down our tool. Note that our approach can also be generalized to more than two reads collections.





\subsection{SNP calling}
\label{ssec:post}

Our SNPs calling approach takes as input $\ebwt(\bigR)$, $\lcp(\bigR)$, and $\gsa(\bigR)$ and outputs SNPs in {\sc KisSNP2} format~\cite{kSNP2biblio}: a fasta file containing a pair of sequences per SNP (one per sample, containing the SNP and its context). 
The SNP calling is divided in two main steps.

\textbf{Build clusters.}  First, we scan $\ebwt(\bigR)$ and $\lcp(\bigR)$, find clusters using Theorem \ref{theorem:LCP}, and store them to file as a sequence of ranges on the \EBWT. In addition, while computing clusters we also apply a threshold of minimum \LCP (by default, 16): we cut clusters' tails containing \LCP values smaller than the threshold. This additional filtering drastically reduces the number of clusters saved to file (and hence memory usage and running time), since the original strategy would otherwise output many short clusters containing small \LCP values corresponding to noise.

\textbf{Call SNPs.} The second step takes as input the clusters file, $\ebwt(\bigR)$, $\lcp(\bigR)$, $\gsa(\bigR)$, and $\bigR$, and processes clusters from first to last as follows:
\begin{enumerate}
	\item We test the cluster's length using the Poisson distribution predicted by Theorem \ref{th:eBWT clustering}; if the cluster's length falls in one of the two  tails at the sides of the distribution (by default, the two tails summing up to $5\%$ of the distribution), then the cluster is discarded; Moreover, due to $k$-mers that are not present in the genome but appear in the reads because of sequencing errors (which introduce noise around cluster length equal to 1), we also fix a minimum value of length for the clusters (by default, 4 letters per sample).

\item In the remaining clusters, we find the most frequent nucleotides $b_1$ and $b_2$ of samples 1 and 2, respectively, and check whether $b_1 \neq b_2$; if so, then we have a candidate SNP: for each sample, we use the \GSA to retrieve the coordinate of the read containing the longest right-context without errors (see explanation after Proposition \ref{prop:condition 2 prob}); moreover, we retrieve, and temporarily store in a buffer, the coordinates of the remaining reads in the cluster. 

\item After processing all events, we scan the fasta file storing $\bigR$ to retrieve the reads of interest (those whose coordinates are in the buffer); for each one of them, we compute a partial assembly of the read prefixes preceding the SNP, for each of the two samples. This allows us to compute a left-context for each SNP (by default, of length 20), and represents a further validation step: if the assembly cannot be built because a consensus cannot be found, then the cluster is discarded.
Note that these left-contexts preceding SNPs (which are actually right-contexts if the cluster is formed by reads from the RC strand) allow us to capture non-isolated SNPs.
\end{enumerate}

\textbf{Complexity} In the clustering step, we process the \EBWT and \LCP and on-the-fly output clusters to disk. 
The SNP-calling step performs one scan of the \EBWT, \GSA, and clusters file to detect interesting clusters, plus one additional scan of the read set to retrieve contexts surrounding SNPs. Both these phases take linear time in the size of the input and do not use disk space in addition to the input and output. 
Due to the fact that we store in a buffer the coordinates of reads inside interesting clusters, this step uses an amount of RAM proportional to the number of SNPs times the average cluster size $\lambda$ times the read length $r$ (e.g. a few hundred MB in our case study of Section~\ref{sec:results}). Notice that our method is very easy to parallelize, as the analysis of each cluster is independent from the others.

\subsection{Validation}

Here we describe the validation tool we designed to measure the sensitivity and precision of 
any tool outputting SNPs in {\sc KisSNP2} format. Note that we output SNPs as pairs of reads containing the actual SNPs plus their contexts (one sequence per sample). This can be formalized as follows: the output is a series of pairs of triples (we call them \emph{calls}) $(L',s',R'),\ (L'',s'',R'')$ where $L'$, $R'$, $L''$, $R''$ are the left/right contexts of the SNP in the two samples, and letters $s'$, $s''$ are the actual variant. Given a {\rm .vcf} file (Variant Call Format) containing the ground truth, the most precise way to validate this kind of output is to check that the triples actually match contexts surrounding true SNPs on the reference genome (used here just for accuracy validation purposes). That is, for each pair in the output calls:
\begin{enumerate}
	\item If there is a SNP $s'\rightarrow s''$ in the .vcf that is surrounded in the first sample by contexts $L', R'$ (or their RC), then $(L',s',R'),\ (L'',s'',R'')$ is a true positive (TP).
	\item Any pair  $(L',s',R'),\ (L'',s'',R'')$ that is not matched with any SNP in the ground truth (as described above) is a false positive (FP).
	\item Any SNP in the ground truth that is not matched with any call is a false negative (FN). 
\end{enumerate}
We implemented the above validation strategy with a (quite standard) reduction of the problem to the 2D range reporting problem: we insert in a two-dimensional grid two points per SNP (from the .vcf) using as coordinates the ranks of its right and (reversed) left contexts among the sorted right and (reversed) left contexts of all SNPs (contexts from the first sample) on the F and RC strands. Given a pair $(L',s',R'),\ (L'',s'',R'')$, we find the two-dimensional range corresponding to all SNPs in the ground truth whose right and (reversed) left contexts are prefixed by $R'$ and (the reversed) $L'$, respectively. If there is at least one point in the range matching the variation $s'\rightarrow s''$, then the call is a TP\footnote{Since other tools such as \DiscoSnpPiu do not preserve the order of samples in the output, we actually check also the variant $s''\rightarrow s'$ and also search the range corresponding to $L''$ and $R''$} (case 1 above; note: to be a TP, a SNP can be found either on the F or on the RC strand, or both). Otherwise, it is a FP (case 2 above). Finally, pairs of points (same SNP on the F/RC strands) that have not been found by any call are marked as FN (case 3 above). 
We repeat the procedure for any other SNP found between the two strings $L's'R'$ and $L''s''R''$ to find non-isolated SNPs.

\subsection{Preliminary Experiments}\label{sec:results}

In order to valuate our method, we  compare \PosCluster with \DiscoSnpPiu,
that is a revisiting of the \DiscoSnp algorithm:
while \DiscoSnp detects both heterozygous and homozygous isolated SNPs from any number of read datasets without a reference genome, \DiscoSnpPiu is designed for detecting and ranking all kinds of SNPs and small indels from raw read set(s).
As shown in \cite{discoSNP2017bioRxiv}, \DiscoSnpPiu performs better than state-of-the-art methods in terms of both computational resources and quality of the results.

\DiscoSnpPiu is composed of several independent tools.
As a preprocessing step, the dBG of the input datasets is built, by also removing erroneous $k$-mers. 
Then, \DiscoSnpPiu detects bubbles generated by the presence of SNPs (isolated or not) and indels and it outputs a fasta file containing the variant sequences ({\sc KisSnp2} module). A final step ({\sc kissreads2}) maps back the reads from all input read sets on the variant sequences, mainly in order to determine the read coverage per allele and per read set of each variant. This module also computes a rank per variant, indicating whether it exhibits discriminant allele frequencies in the datasets. 
The last module generates a {\rm .vcf} of the predicted variants. If no reference genome is provided this step is a change of format from fasta to {\rm .vcf} ({\sc VCFcreator} module).

We propose two experiments simulating two human chromosomes haploid read sets obtained mutating (with real {\rm .vcf} files) real reference chromosomes\footnote{\url{ftp.1000genomes.ebi.ac.uk//vol1/ftp/technical/reference/phase2_reference_assembly_sequence/hs37d5.fa.gz}}.
The final goal of the experiments is to reconstruct the variations contained in the original (ground truth) {\rm .vcf} files.
We generated the mutated chromosomes using the 1000 genome project (phase $3$) {\rm .vcf} files\footnote{\url{ftp.1000genomes.ebi.ac.uk/vol1/ftp/release/20130502/}} related to chromosomes $16$ and $22$, suitably filtered to keep only SNPs of individuals HG00100 (ch.$16$) and HG00096 (ch.$22$).
From these files, we simulated Illumina sequencing with SimSeq \cite{SimSeq2011}, both for reference and mutated chromosomes: individual HG00096 (ch.$22$) at a $29$x getting $15,000,000$ of $100$-bp reads, and individual HG00100 (ch.$16$) a $22$x getting $20,000,000$ of $100$-bp reads.


Our framework has been implemented in C++ and is available at 
\url{https://github.com/nicolaprezza/eBWTclust}. 
All tests were done on a DELL PowerEdge R630 machine, used in non exclusive mode.
Our platform is a $24$-core machine with Intel(R) Xeon(R) CPU E5-2620 v3 at $2.40$ GHz, with $128$ GB of shared memory. The system is Ubuntu 14.04.2 LTS. Note that, unlike \DiscoSnpPiu, our tool is currently able to use one core only.

We experimentally observed that the pre-processing step is more computationally expensive than the actual SNP calling step. The problem of computing the \EBWT is being intensively studied, and improving its efficiency is out of the aim of this paper. However, a recent work \cite{Dolle2017} suggests that direct storing of read data with a compressed \EBWT leads to considerable space savings, and could therefore become the standard in the future. Our strategy can be easily adapted to directly take as input these compressed formats. 
For both tools, we thus omit time/space requirements of the preprocessing steps: computing the data structures described in Section \ref{subsec:Pre-processing} for \PosCluster, and constructing the dBG and removing erroneous $k$-mers for \DiscoSnpPiu. 
Building the dBG requires a few minutes and, in order to keep the RAM usage very low, no other information other than $k$-mer presence is stored in the dBG used by \DiscoSnpPiu. On the other hand, the construction of \EBWT, \LCP and \GSA can take a few hours (around 90 minutes by using \EGSA \cite{Louza2017}). So, overall \DiscoSnpPiu is faster than \PosCluster when considering both pre-processing and post-processing.

We run \DiscoSnpPiu with default parameters (includes $k$-mers size $31$) except for $P=3$ (it searches up to $P$ SNPs per bubble) and $b$ ($b=0$ forbids variants for which any of the two paths is branching;
$b=2$ imposes no limitation on branching; $b=1$ is inbetween).

\PosCluster takes as input few main parameters, among which the most important are the lengths of right and left contexts surrounding SNPs in the output (\texttt{-L} and \texttt{-R}), the minimum cluster size (\texttt{-m}), and (\texttt{-v}) the maximum number of non-isolated SNPs to seek in the left contexts (like parameter $P$ of \DiscoSnpPiu). 
We decided to output 20 nucleotides preceding (and including) the SNP (\texttt{-L 20}), 30 following the SNP (\texttt{-L 30}), a minimum cluster size of \texttt{-m 6} and \texttt{-m 4}, and \texttt{-v 3}.


In Table \ref{table_experiments},
we show the number of TP, FP and FN as well as sensitivity (SEN), precision (PREC), and the number of non-isolated SNPs found by the tools. The outcome is that \PosCluster is always more precise and sensitive than \DiscoSnpPiu.
Moreover, while in our case precision is stable and always quite high ($\approx98.5$\% with \texttt{-m 6} for HG00096 and \texttt{-m 4} for HG00100, and $>99$\% with \texttt{-m 6} HG00100), 
for \DiscoSnpPiu precision is much lower in general, and even drops with $b=2$, especially with lower coverage, when inversely sensitivity grows. Sensitivity of \DiscoSnpPiu gets close to that of \PosCluster only in case $b=2$, when its precision drops and memory and time get worse than ours. Note that precision and sensitivity of \DiscoSnpPiu are consistent with those reported in \cite{discoSNP2017bioRxiv}.




\begin{table}[th!]
\begin{center}
\begin{tabular}{|c@{\ }|c@{\ }|c@{\ }|c@{\ }|c@{\ }|c@{\ }|c@{\ }|c@{\ }|c@{\ }|c@{\ }|}
\multicolumn{10}{c}{Individual HG00096 vs reference (chromosome 22, $50818468$bp), coverage $29\times$ per sample}     \\
\hline
\hline	
 Tool       	& Param.	& Wall 	 	& RAM	& TP           	& FP &   FN  		& SEN  		& PREC & Non-isolated\\ 
       		&	 	& Clock		& in MB	 &                	&      &          		&         		&            &  SNP\\ 
\hline
\DiscoSnpPiu &   b=0	& 5:07	& 101	& 	32773 	& 	3719 & 13274	 	& 	 71.17\%	& 	89.81\% & 4707/8658	\\ 
\hline
 \DiscoSnpPiu &   b=1	& 16:39	& 124	&  37155	& 10599	 & 	 8892	&  80.69\% 	& 77.80\% & 5770/8658	\\ 
\hline
 \DiscoSnpPiu &   b=2	& 20:42	& 551	& 40177	 	& 58227	 & 	5870 	& 87.25\% 	& 40.83\% &	6325/8658 \\ 
 \hline 
\PosCluster &  m=4      & 19:43		& 415	& 42973 & 2639 &  3074 & 93.32\%       & 94.21\% & 7268/8658\\ 
\hline             
\PosCluster &  m=6      & 24:58		& 411	& 41972 & 630 & 4075  & 91.15\%       & 98.52\% & 6940/8658\\
\hline
\end{tabular}
\qquad
\begin{tabular}{|c@{\ }|c@{\ }|c@{\ }|c@{\ }|c@{\ }|c@{\ }|c@{\ }|c@{\ }|c@{\ }|c@{\ }|}
\multicolumn{10}{c}{Individual HG00100 vs reference (chromosome 16, $90338345$bp), coverage $22\times$ per sample}     \\
\hline
\hline	
 Tool       	& Param.	& Wall 	 	& RAM	& TP           	& FP &   FN  		& SEN  		& PREC & Non-isolated\\ 
       		&	 	& Clock		& in MB	 &                	&      &          		&         		&            &  SNP\\ 
        \hline
\DiscoSnpPiu &   b=0	& 6:20	& 200	& 	48119 	& 10226	 & 	 18001	& 	 72.78\%	& 	82.47\% & 6625/11055	\\ 
\hline
\DiscoSnpPiu &   b=1	& 31:57	& 208	& 53456 	& 24696	 & 	12664 	&  80.85\% 	& 68.40\% & 7637/11055	\\ 
\hline
\DiscoSnpPiu &   b=2	& 51:45	& 1256	& 57767	 	& 124429	 & 	8353 	& 87.37\% 	& 31.71\% &	8307/11055 \\ 
\hline
\PosCluster  & m=4 		& 41:12	& 423		& 61264 & 943 & 4856  & 92.66\%       & 98.48\% & 9314/11055\\ 
\hline
\PosCluster  & m=6 		& 43:51 & 419 & 58085  & 391 & 8035  & 87.85\%       & 99.33\% & 8637/11055 \\ 
\hline
\end{tabular}
\end{center}
\caption{Comparative results of \PosCluster (only SNP calling) and \DiscoSnpPiu (only {\sc KisSnp2} and {\sc kissreads2}). Wall clock is the elapsed time from start to completion of the instance, while RAM is the peak Resident Set Size (RSS). Both values were taken with \texttt{/usr/bin/time} command.}
\label{table_experiments}
\end{table}

\section{Conclusions}

We introduced a positional clustering framework for the characterization of breakpoints in the \EBWT, paving the way to several possible applications in assembly-free and reference-free analysis of NGS data. The experiments proved the feasibility and potential of our approach. 
Further work will focus on improving the prediction in highly repeated genomic regions and using our framework to predict SNPs, predict INDELs, haplotyping, correcting sequencing errors, detecting Alternative Splicing events in RNA-Seq data, and sequence assembly.






\bibliography{biblio}

\end{document}